\documentclass[reqno,11pt]{amsart}
\usepackage{amsmath,amsfonts,amsgen,amstext,amsbsy,amsopn,amsthm, amssymb}

 \newtheorem{Theorem}{THEOREM}
 
 \theoremstyle{definition}

\newcommand\nn\nonumber

 \newcommand{\R}{\mathbb{R}}

 \newcommand\infspec{{\rm{inf\, spec\,}}}

 \DeclareMathOperator{\const}{const}

\begin{document}

\title[Absence of bound states implies non-negativity of the
scattering length]{Absence of bound states implies non-negativity \\
  of the scattering length}

\author[R. Seiringer]{Robert Seiringer} \address{Department of Mathematics and Statistics, McGill
  University, 805 Sherbrooke Street West, Montreal, QC H3A 2K6,
  Canada} 
\email{robert.seiringer@mcgill.ca}

\date{April 26, 2012}

\begin{abstract}
  We show that bosons interacting via pair potentials with negative
  scattering length form bound states for a suitable number of
  particles. In other words, the absence of many-particle bound states
  of any kind implies the non-negativity of the scattering length of
  the interaction potential.
\end{abstract}

\thanks{\copyright\, 2012 by  
  the author. This paper may be reproduced, in its entirety, for
  non-commercial purposes.}

\maketitle

\phantom{x}

While there is a huge literature on two-particle bound states in
quantum mechanics, relatively little is known concerning bound states
of more than two particles. Via the introduction of center-of-mass and
relative coordinates, the former question reduces to a one-particle
problem, i.e., the spectral analysis of an operator of the form
$-\Delta + V(x)$, with $\Delta$ denoting the Laplacian and $V$ a
potential, i.e., a multiplication operator, which decays at
infinity. Also for more than two particles the center-of-mass motion
can be separated, but the resulting operator is more complicated, and
the potential does not decay in certain directions.

In this article, we shall show that if the pair interaction potential
has a negative scattering length, there always exist bound states for
a suitable number of particles. At low density and low temperature
such a system will thus not behave like an atomic gas, but will form
molecules or even much larger bound clusters. We note that the atoms
used in current experiments on Bose-Einstein condensation of
ultra-cold dilute gases have scattering length of either sign (see,
e.g., \cite{DGPS,BDZ}).

For $N\geq 2$, we consider the Hamiltonian
$$
H_N = - \sum_{i=1}^N \Delta_i  + \sum_{1\leq i<j \leq N} v(x_i-x_j)
$$
on $L^2(\R^{3N})$, with $\Delta$ the usual Laplacian on $\R^3$. We
assume that $v\in L^1(\R^3)$ is real-valued and radial, and that $v_-
\in L^{3/2}(\R^3)$, where $v_-$ denotes the negative part of $v$,
i.e., $v_-(x)=\max\{0,-v(x)\}$. Under these assumptions, the quadratic
form defined by $H_N$ is bounded from below, and hence $H_N$ gives
rise to a self-adjoint operator via the Friedrichs extension \cite{RS2}.  

Let $E_N = \infspec H_N$ denote the ground state energy of $H_N$, and let 
$$
\Sigma_N = \min_{1\leq n \leq N-1} \left\{ E_{N-n} + E_n\right\}
$$
denote the lowest energy of two separate clusters of particles. Note that $E_N \leq \Sigma_N \leq E_{N-1}$. The
HVZ Theorem \cite{RS4} implies that if $E_N<\Sigma_N$, then the
$N$-particle system has a bound state, i.e., the Hamiltonian $H_N$ has,
after removal of the center-of-mass motion, an eigenvalue at the
bottom of its spectrum. We note that $E_N$ is always attained in the
sector of permutation-invariant functions \cite{LS}, hence we are effectively
dealing with bosons even though we defined $H_N$ on the full space
$L^2(\R^{3N})$ for simplicity.

The {\em scattering length} of the interaction potential $v$ is defined as 
$$a := \lim_{R\to \infty} a_R
$$
where $a_R$ is given by the variational principle \cite{LY,LSSY}
\begin{equation}\label{def:ar}
a_R = \frac 1 {8\pi} \inf\left\{ \int_{|x|\leq R} \left( 2 |\nabla f(x)|^2 + v(x)|f(x)|^2 \right) dx \, : \,  \ f(x) = 1 \ \text{for $|x|=R$} \right\}\,.
\end{equation}
The infimum is over all $H^1$-functions on the ball $\{|x|\leq
R\}$ satisfying the boundary condition $f(x)=1$ on the boundary of the
ball. Under the assumption that $H_2\geq 0$, i.e., the absence of
two-particle bound states, the existence of a minimizer for
(\ref{def:ar}) was shown in \cite[Appendix~A]{LY}. Note that since
$a_{R_1} \leq a_{R_2} + (8\pi)^{-1}\int_{R_1\leq |x|\leq R_2} v(x) dx$
for $R_2>R_1$, the limit $\lim_{R\to \infty} a_R$ always exists for
$v\in L^1(\R^3)$, but it could equal $-\infty$.  It is easy to see that
finiteness of $a$  implies that $H_2\geq 0$, i.e.,
\begin{equation}\label{ah2}
a> -\infty \quad  \Rightarrow \quad  H_2\geq 0 \,.
\end{equation}

Eq.~(\ref{def:ar}) is the correct definition of the scattering length
only in the absence of two-particle bound states. The scattering length can be defined by
other means even in the presence of two-particle bound states (see,
e.g., \cite{HS}), but it is not given by the above variational
principle in this case. Since our main concern here is the case of
$H_2\geq 0$, we find it convenient to work with the
definition~(\ref{def:ar}), however.

Our main result is the following.

\begin{Theorem}\label{thm}
 If $H_N\geq 0$ for all $N\geq 2$, then $a\geq 0$. 
\end{Theorem}

As discussed above, finiteness of $a$ implies that $E_2=0$. If $a<0$,
Theorem~\ref{thm} implies the existence of an $N\geq 3$ with
$\Sigma_N = 0$ but $E_N<0$.  In other words, negativity of the scattering length
of $v$ implies the existence of bound states of some kind.
We note that the converse of
Theorem~\ref{thm} does not hold. There exist interaction potentials
$v$ with positive scattering length such that $H_2\geq 0$ but $H_3$
has negative spectrum, i.e., with three-particle bound states but no
two-particle bound states \cite{B}.

For large, negative $a$, the existence of three-particle bound
states is well known \cite{Ef,Yaf,OS,Sob,SRZ}. Our Theorem~\ref{thm} implies that
$a$ does not actually have to be large for bound states to exist. As
long as it is negative, bound states exist, but they might require
more than three particles.

We note that if $\int_{\R^3} v(x) dx < 0$, then $a<0$. In this case,
it is not difficult to see that there exist bound states for an
infinite sequence of particle numbers. This follows from the fact that
$E_N \sim - \const N^2$ for large $N$, while obviously $E_N \sim
\const N$ if bound states exist only for finitely many particle
numbers. It remains an open problem to decide whether bound
states exist for an infinite sequence of particle numbers for all
interaction potentials with $a<0$.

Theorem~\ref{thm} can be extended to dimensions larger than three. It
does not extend to one and two dimensions, however. The corresponding
expression (\ref{def:ar}) in one or two dimensions is always positive
in the absence of two-particle bound states \cite{LY,LSSY}. Another
important difference concerns the fact that arbitrarily shallow
negative potentials lead to the existence of two-particle bound states
in one and two dimensions, which is not the case in three and more
dimensions.

In the remainder of this paper, we shall give the proof of Theorem~\ref{thm}.

\begin{proof}[Proof of Theorem~\ref{thm}]
  We shall assume that $a<0$, and show that for large enough $N$ we
  can find a $\Psi \in L^2(\R^{3N})$ with $\langle \Psi | H_N
  \Psi\rangle < 0$. Since $a<0$ and $v\in L^1(\R^3)$, we can find a $R>0$
  such that $a_R + (8\pi)^{-1}\int_{|x|>R} v(x) dx < 0$. In particular, there
  exists a non-negative function $\varphi : \R_+\to \R_+$ with
  $\varphi(t) = 1$ for $t\geq R$ such that
\begin{equation}\label{def:b}
b:=  \int_{\R^3} \left( 2  \varphi'(|x|)^2 + v(x)\varphi(|x|)^2 \right) dx < 0\,.
\end{equation}
For instance, one can choose $\varphi$ to be the minimizer of
(\ref{def:ar}) for $|x|\leq R$. By an approximation argument, we can
assume $\varphi$ to be bounded without loss of generality. That is,
for some $c\geq 0$, we have
$$
\varphi(t)^2 \leq 1 + c\, \theta(R-t)
$$
for all $t\geq 0$, with $\theta$ denoting the Heaviside step function
$$
\theta(t) = \left\{ \begin{array}{cc} 1 & \text{for $t\geq 0$} \\ 0 & \text{for $t<0$.} \end{array}\right.
$$
 
For $g\in C_0^\infty(\R^3)$ real-valued and with $\int_{\R^3} g(x)^2 dx =1$, we shall consider the function
$$
\Psi(x_1,\dots,x_N) = \varphi( t ) \prod_{k=1}^N g(x_k)
$$
where
$$
t= t(x_1,\dots,x_N) = \min_{1\leq i<j\leq N} |x_i-x_j|\,.
$$
Note that $t$ is Lipschitz continuous and differentiable almost everywhere. A simple calculation yields
$$
\sum_{k=1}^N |\nabla_k \varphi(t)|^2 = 2 \varphi'(t)^2\,.
$$

An integration by parts shows that 
\begin{align}\nonumber
 \sum_{k=1}^N \int_{\R^{3N}} \left| \nabla_k \Psi \right|^2 dx_1 \cdots dx_N  & = 2 \int_{\R^{3N}} \varphi'(t)^2 \prod_{j=1}^N g(x_j)^2 dx_j \\ & \quad -  N \int_{\R^{3N}} g(x_1) \Delta g(x_1)  \varphi(t)^2 dx_1 \prod_{j=2}^N g(x_j)^2 dx_j \,. \label{f2t}
\end{align}
Here we have also used the permutation-symmetry of $\Psi$ in the last term, in order to replace the sum over $k$ by just one summand, multiplied by $N$. 
We can bound 
$$
\varphi'(t)^2 \leq \sum_{i<j} \varphi'(|x_i-x_j|)^2
$$
and hence the first term on the right side of (\ref{f2t}) is bounded above by 
$$
N(N-1) \int_{\R^6} g(x)^2 \varphi'(|x-y|)^2 g(y)^2 \,dx\,dy\,.
$$
To bound the second term, we simply use $\varphi(t)^2 \leq 1 + c$. 
We shall also drop the positive part of $g(x)\Delta g(x)$, and arrive at the upper bound
\begin{align}\nonumber
 \sum_{k=1}^N \int_{\R^{3N}} \left| \nabla_k \Psi \right|^2 dx_1 \cdots dx_N   & \leq N(N-1) \int_{\R^6} g(x)^2 \varphi'(|x-y|)^2 g(y)^2 \,dx\,dy \\ & \quad  + (1+c) N \| [g\Delta g]_- \|_1 \,, \label{ll0}
\end{align}
where $[s]_- = \max\{0,-s\}$ denotes the negative part.

Next we shall investigate the expectation value of the interaction
potential. Because of the permutation-symmetry of $\Psi$, it suffices
to consider $\langle\Psi| v(x_1-x_2) \Psi\rangle$, multiplied 
by $N(N-1)/2$, the number of pairs of particles. To bound this term,
we shall use the fact that
\begin{align}\nonumber
&\left| \varphi(t)^2 - \varphi(|x_1-x_2|)^2\right| \\ & \leq (1+c) \left( \sum_{j=3}^N \theta(R-|x_1-x_j|) + \sum_{k=2}^N \sum_{j=k+1}^N \theta(R-|x_k-x_j|) \right)\,. \label{hol}
\end{align}
To see the validity of this bound, note that the left side is at most
$1+c$. The right side is at least $1+c$, unless all the pairs $(x_k,x_j)$ with
$\{k,j\}\neq \{1,2\}$ are separated a distance larger than $R$, in which case it equals zero. Also the left side is then zero, however; either $|x_1-x_2|\geq R$, in which case it is zero since also $t\geq R$ and $\varphi(s)=1$ for $s\geq R$, or $|x_1-x_2|<R$, in which case $t=|x_1-x_2|$. In any case, (\ref{hol}) holds.

Using (\ref{hol}), we have
\begin{align}\nonumber
&\left\langle\Psi\left| v(x_1-x_2) \right. \!\Psi\right\rangle \\ \nonumber & \leq \int_{\R^6} g(x)^2 v(x-y) \varphi(|x-y|)^2 g(y)^2 \, dx\, dy \\ & \nonumber \quad  + (1+c) \frac{(N-2)(N-3)}{2} \left(\int_{\R^6} g(x)^2 |v(x-y)| g(y)^2 \, dx\, dy \right) \\ \nonumber & \quad \qquad\qquad\qquad\qquad\qquad\qquad \times \left( \int_{\R^6} g(x)^2 \theta(R-|x-y|) g(y)^2 \, dx\, dy  \right) \\ &  \quad + 2 (1+c) (N-2) \int_{\R^9} g(x)^2 |v(x-y)| g(y)^2 \theta(R-|y-z|) g(z)^2 \, dx\, dy\, dz \label{ll} \,.
\end{align}
The terms on the second line on the right side can be bounded with the aid of a Schwarz inequality as
\begin{equation*}\label{schw2}
\int_{\R^6} g(x)^2 |v(x-y)| g(y)^2 \, dx\,dy \leq  \|v\|_1 \|g\|_4^4 
\end{equation*}
and
\begin{equation*}\label{schw1}
\int_{\R^6} g(x)^2 \theta(R-|x-y|) g(y)^2 \, dx\,dy\leq \frac {4\pi}3 R^3 \|g\|_4^4 \,.
\end{equation*}
Similarly, we can bound the integral in the last line of (\ref{ll}) as 
$$
  \int_{\R^9} g(x)^2 |v(x-y)| g(y)^2 \theta(R-|y-z|) g(z)^2 \, dx\, dy\, dz \leq \frac {4\pi}3 R^3  \|v\|_1 \|g\|_\infty^2 \|g\|_4^4\,.
$$
As a result, we have thus shown that
\begin{align}\nonumber
\left\langle\Psi\left| v(x_1-x_2) \right. \!\Psi\right\rangle  & \leq \int_{\R^6} g(x)^2 v(x-y) \varphi(|x-y|)^2 g(y)^2 \, dx\, dy \\ & \quad  + (1+c)\frac{2\pi}{3}  (N-2) R^3 \|v\|_1  \left(  (N-3) \|g\|_4^8   + 4 \|g\|_\infty^2 \|g\|_4^4 \right) \label{ll2} \,.
\end{align}

By combining the bounds (\ref{ll0}) and (\ref{ll2}), we obtain
\begin{align*}
\left\langle \Psi\left| H_N\right.\!\Psi\right\rangle & \leq \frac{N(N-1)}2 \int_{\R^6} g(x)^2 h(x-y) g(y)^2 \, dx \, dy \\ & \quad + (1+c) N \left( \| [g\Delta g]_- \|_1   +\frac{\pi}{3} N^2 R^3 \|v\|_1  \|g\|_4^4 \left( N \|g\|_4^4  + 4  \|g\|_\infty^2 \right)\right)
\end{align*}
where $h$ denotes the function
$$
h(x) = 2 \varphi'(|x|)^2 + v(x) \varphi(|x|)^2 \,.
$$
Note that $\int_{\R^3} h(x) dx = b < 0$ by (\ref{def:b}).

Let $g_0\in C_0^\infty(\R^3)$ be real-valued and $L^2$-normalized. For $L>0$,  we shall choose 
$$
g(x) = L^{-3/2} g_0(x/L)\,.
$$
Then $g$ is also $L^2$-normalized, and satisfies $\|g\|_4 = L^{-3/4} \|g_0\|_4$,
$\|g\|_\infty = L^{-3/2} \|g_0\|_\infty$, and $\| [g\Delta g]_-\|_1 =
L^{-2}\| [g_0\Delta g_0]_-\|_1$. Denoting the various constants
collectively by $C$, we thus have
\begin{equation}
\left\langle \Psi\left| H_N\right.\!\Psi\right\rangle  \leq \frac{N(N-1)}2 \int_{\R^6} g(x)^2 h(x-y) g(y)^2 \, dx \, dy  + C \frac{N}{L^2} \left( 1  + \|v\|_1 \frac{N^3 R^3}{L^4} \right)  \,. \label{fr}
\end{equation}
The first term on the right side equals
$$
\frac{N(N-1)}{2L^6} \int_{\R^6} g_0(x/L)^2 h(x-y) g_0(y/L)^2 \, dx \, dy \,.
$$
Since $h$ is an $L^1$-function, we can use dominated convergence in Fourier space to conclude that 
$$
\lim_{L\to \infty} L^{-3} \int_{\R^6} g_0(x/L)^2 h(x-y) g_0(y/L)^2 \, dx \, dy = \int_{\R^3} h(x) dx \,\|g_0\|_4^4 = b \|g_0\|_4^4 < 0\,.
$$
For all $L$ large enough, we can thus bound 
$$
\frac{N(N-1)}{2L^6} \int_{\R^6} g_0(x/L)^2 h(x-y) g_0(y/L)^2 \, dx \, dy \leq \frac b 4 \frac{N(N-1)}{L^3} \|g_0\|_4^4\,.
$$
If we choose $L \ll N \ll L^{3/2}$, the last term in (\ref{fr}) is much smaller than $N^2/L^3$ for large $L$. For this choice of $N$, we thus have $\langle \Psi|H_N \Psi\rangle <0$ for large enough $L$. This concludes the proof.
\end{proof}
\bigskip

\noindent {\it Acknowledgments.} The author thanks Wilhelm Zwerger for
interesting comments. Partial financial support by NSERC is gratefully
acknowledged.

\bigskip


\end{document}